\documentclass{article}
\usepackage[utf8]{inputenc}
\usepackage[a4paper, total={6in, 9in} ]{geometry}
\usepackage{braket}
\usepackage{xcolor}
\usepackage{amsmath}
\usepackage{amsfonts}
\usepackage{amsthm}
\usepackage{amssymb}
\usepackage{hyperref}
\usepackage{cleveref}
\usepackage{graphicx}
\usepackage{svg}
\usepackage{float}
\usepackage{tikz}
\usetikzlibrary{patterns, shapes.arrows}
\usepackage{adjustbox}
\usepackage{tikz-network}
\usepackage[linesnumbered]{algorithm2e}
\usepackage{multicol}
\usepackage[backend=biber,style=alphabetic,sorting=ynt]{biblatex}
\usepackage{pgfplots}
\usepackage{etoc}
\usepackage{doi}
\usepackage[hyphenbreaks]{breakurl}
\DeclareUnicodeCharacter{2212}{−}
\usepgfplotslibrary{groupplots,dateplot}
\pgfplotsset{compat=newest}

\newtheorem{theorem}{Theorem}

\newtheorem{claim}{Claim}

\newtheorem*{theorem*}{Theorem}

\crefname{lemma}{Lemma}{Lemmas}
\hypersetup{colorlinks=true}

\newcommand{\onotation}[1]{\(\mathcal{O} \left( {#1}  \right) \)}
\newcommand{\ona}[1]{\onotation{#1}}

\begin{document}

\title{No-Existence Of General Diffusion.}
\author{David Ponarovsky}
\maketitle

\begin{abstract}\textit{We show that given two arbitrary states $\ket{\psi},\ket{\phi}$ it is impossible to compute the transformation: $ \ket{\psi}\ket{\phi} \mapsto \ket{\psi}\left( \mathbb{I} - 2 \ket{\psi}\bra{\psi} \right)\ket{\phi} $ The contradiction of the existence of such operator follows by showing that using it, two players can compute the disjoints of their sets in a single round and $O\left( \sqrt{n} \right)$ communication complexity, which shown by Braverman to be impossible \cite{Braverman}. }
\end{abstract}

\section{Preamble} It's widely believed that quantum machines have a significant advantage over classical optimization tasks. Simple algorithms, which could be interpreted as the quantum version of ''scanning all the options'', cut the running time by the square root of the classical magnitude. That cut is achieved by using the superposition principle most straightforwardly known as the Amplitude Amplification algorithm \cite{Brassard_2002}, \cite{grover1996fast}. 

General speaking, this method transforms a known state $\ket{\psi}$  with probability $a$ to measure $\ket{i}$ to a state in which the desired measurement obtained with probability greater than $\frac{1}{2}$ at the cost of less than $\sqrt{a}$ Grover iterations. Using this process, One can initialize a uniform distribution over $n$ elements and amplify the probability to measure a desired state at $\sqrt{n}$ time. To understand the power gained by this method, we mention max extraction as a use-case \cite{ahuja1999quantum}. While any classical algorithm which runs at square root time scans at most $ \Theta( \sqrt{n} )$ elements and might miss the maximum with probability at least $ 1 - \Theta(1/\sqrt{n})$. Therefore can't yield a constant probability to sample the maximum element, Quntemly, this limitation doesn't hold. And the gap amplification indeed enables a square root time maximum extraction algorithm. 
  
  A critical requirement for that procedure is to have the ability to generate copies of the initial state, Formulated by \cite{Brassard_2002} as holding an algorithm $\mathcal{A}$, which does not make any measurements, such that $\mathcal{A}\ket{0}=\ket{\Psi}$. Assuming having this ability, one could mimic the scattering done in the Grover search but restrict himself to be supported on $\ket{\Psi}$.

  One question that might arise is whether the above amplification process can be done assuming nothing but having a single entity of the initial state. Both positive and negative answers will illuminate the fundamentals behind transferring probability weight. We partially answered that question by proving that the given copy alone cannot simulate the diffusion step. We formulate the above by the following theorem: 

  \begin{theorem} \textit{ There is no operator $D$ that for given two arbitrary states $\ket{\psi},\ket{\phi}$ compute the transformation:} 
\begin{equation*}
    D \ket{\psi}\ket{\phi} = \ket{\psi}\otimes\left( \mathbb{I} - 2 \ket{\psi}\bra{\psi} \right)\ket{\phi} 
\end{equation*}
\end{theorem}

We name the gate above the \textit{General Diffusion} gate. If such a gate existed, it could be used as the projection operator to simulate the amplitude amplification procedure. The contradiction of the existence follows by showing that using $D$, two players can compute the disjoints of their sets in a single round and $O\left( \sqrt{n} \right)$ communication complexity contradicts the fact that $r$-rounds two-party computation needs at least $\Omega\left( \frac{n}{r} \right)$ communication to compute disjoined (up to log factors) \cite{Braverman}.    

\paragraph{Quantum Communication Complexity Of Disjointness.}
Consider the following communication problem.
As inputs Alice gets an \(x\), and Bob gets a \(y\), where \(x, y \in \{0, 1\}^n \), and by exchanging information, they want to determine if there is an index \(k\) such that \(x_k = y_k = 1 \) or not. 
In other words, if \(x\) encodes the set \(A = \{k | x_k = 1\} \), and \(y\) encodes \(B = \{k | y_k = 1\}\), 
then Alice and Bob want to determine whether \( A \cap B \) is empty.

The classical randomized communication complexity of this problem is \ona{n} \cite{v003a011}.
Assuming Alice and Bob can exchange quantum messages, It is known that Alice and Bob can solve the task
correctly with probability greater than \(2/3\) by exchanging at most \ona{\sqrt{n}\log n } qubits 

\section{The Reduction.} 
Assume by way of contradiction the existence of $D$ defined above.  
Let \( x^{(j)} \) be the \(j\)-th \(\sqrt{n}\)-block of \(x\), e.g \(x^{(j)} = x_{j\sqrt{n}},x_{j\sqrt{n}+1}...,x_{(j+1)(\sqrt{n})-1}  \). And denote by \( \ket{\psi_x} \in \mathcal{H}_{2}^{\bigotimes \sqrt{n}} \bigotimes \mathcal{H}_{\sqrt{n}} \) the uniform superposition state over the \( x^{(j)}\)-'s "tensored" with \(\sqrt{n}\)-qudit (which will correspond to the block number). 
\[ \ket{\psi_x} = \frac{1}{n^\frac{1}{4}}\sum_{j}^{\sqrt{n}}\ket{x^{(j)}}\ket{j} \] Note that the encoding of \( \ket{\psi_x} \) require only \( \sqrt{n} + \log(\sqrt{n}) \) qubits.
Clearly, both Alice and Bob can generate the states \( \ket{\psi_x}, \ket{\psi_y} \), then Bob sends his share to Alice.
We know that there is a classical circuit with logarithmic depth in \( \sqrt{n} \) that act over the pure states \( \ket{x^{(j)}}\ket{j}, \ket{y^{(k)}}\ket{k} \) and decides whether \[ \left( j =  k \right) \ \bigwedge  \ \left( \bigvee_{i \in [ \sqrt{n} ] } x^{(j)}_{i} \ \wedge \  y^{(k)}_{i} \right)   \]

Denote it by \( C \) and by \( U \) the phase flip controlled by $C$ i.e. $U\ket{i}=\left( -1 \right)^{C\left( i \right)}\ket{i}$. 

The following claim argues that $D, U$ are sufficient for Alice to simulate a single iteration of the amplitude amplification. Since the technical details of the amplification procedure are not the focus of this paper, we only show equivalence without defining the operators, and the notation used by \cite{Brassard_2002}.   

\begin{claim} \textit{ Recall the operator $\mathbf Q  = - {\mathcal A}  {\mathbf S}_0 
  {\mathcal A}^{-1}  {\mathbf S}_\chi$ defined in \cite{Brassard_2002}, such that $ \mathcal A \ket{0} = \ket{\Psi} = \ket{\psi_{x}}\ket{\psi_{y}}$ and 
consider the generalize diffusion gate $D$, Denote by $\mathcal{H}_{\Psi}$ the space which is spanned by the $\ket{\Psi}$ support. Then it holds that for any state $ \ket{\phi} \in \mathcal{H}_{\Psi} $:}
\begin{equation*}
  \left(  \mathbb{I} \otimes \mathbf Q \right) \ket{\psi_{x}}\ket{\psi_{y}} \ket{\phi} =  - D \left( \mathbb{I} \otimes U \right)  \ket{\psi_{x}}\ket{\psi_{y}} \ket{\phi} 
\end{equation*}
\end{claim}
\begin{proof} Let $\ket{\Psi_0}, \ket{\Psi_1}$ be the base which span $ \mathcal{H}_{\Psi}$ and in addition $U\ket{\Psi_0} = \ket{\Psi_0}, U\ket{\Psi_1} =- \ket{\Psi_1}$.

First consider the case in which the dimension of $ \mathcal{H}_{\Psi}$ is exactly 1, If $ \ket{\Psi} $ supported only on non-satisfying states (i.e $\ket{\Psi} = \ket{\Psi_{0}}) $ then it's clear that $ I \otimes U $ act over the $ \ket{\Psi}\ket{\Psi} $ as identity and therefore $ -D\left( I \otimes U \right) $ act also as identity: 
\begin{equation*}
  -D\left( I \otimes U \right) \ket{\Psi}\ket{\Psi} = -\ket{\Psi}\left( I - 2\ket{\Psi}\bra{\Psi}  \right) \ket{\Psi} = \ket{\Psi}\ket{\Psi}
\end{equation*}
Similar calculation yields that the action is trivial also when  $ \mathcal{H}_{\Psi}$  supported only over $ \ket{\Psi_1} $.  

\paragraph{}

It is left to show the equivalence when $\ket{\Psi}$ supported both over $\ket{\Psi_0}$ and $\ket{\Psi_1}$. Then it follows that:

    \begin{equation*}
      \begin{split}
    - D & \left( \mathbb{I} \otimes U \right)  \ket{\psi_{x}}\ket{\psi_{y}} \ket{\Psi_1} =    D  \ket{\psi_{x}}\ket{\psi_{y}} \ket{\Psi_1} \\
  & = \ket{\psi_{x}}\ket{\psi_{y}} \left( \mathbb{I} - 2 \ket{\psi_{x}}\ket{\psi_{y}} \bra{ \psi_{x}}\bra{\psi_{y}} \right) \ket{\Psi_1} \\
  & =  \ket{\psi_{x}}\ket{\psi_{y}} \left( \mathbb{I} - 2 \ket{\Psi} \bra{\Psi} \right) \ket{\Psi_1}  \\ 
  & = \ket{\psi_{x}}\ket{\psi_{y}} \left( \left( 1 - 2a  \right)\ket{\Psi_1} - 2a \ket{\Psi_0} \right) \\ 
  & \\ 
  - D & \left( \mathbb{I} \otimes U \right)  \ket{\psi_{x}}\ket{\psi_{y}} \ket{\Psi_0} =   - D  \ket{\psi_{x}}\ket{\psi_{y}} \ket{\Psi_0} \\
  & = - \ket{\psi_{x}}\ket{\psi_{y}} \left( \mathbb{I} - 2 \ket{\psi_{x}}\ket{\psi_{y}} \bra{ \psi_{x}}\bra{\psi_{y}} \right) \ket{\Psi_0} \\
  & = - \ket{\psi_{x}}\ket{\psi_{y}} \left( \mathbb{I} - 2 \ket{\Psi} \bra{\Psi} \right) \ket{\Psi_0} \\ 
  & = - \ket{\psi_{x}}\ket{\psi_{y}} \left( \left(  - (2-2a)  \right)\ket{\Psi_1} + 1 -(2 - 2a) \ket{\Psi_0} \right) \\ 
  & = \ket{\psi_{x}}\ket{\psi_{y}} \left( \left( 2-2a \right)\ket{\Psi_1} + \left(1  - 2a\right) \ket{\Psi_0} \right) 
\end{split}
\end{equation*}
\end{proof}
Now, it's clear that Alice could simulate the \textbf{algqsearch} algorithm \cite{Brassard_2002}, 

\paragraph{Theorem 3.} \textit{Quadratic speedup without knowing $\mathbf{a}$
There exists a quantum algorithm \textbf{algqsearch} with the following property.
Let $\mathcal A$ be any quantum algorithm that uses no measurements,
and let $\chi: \mathbb{N}  \rightarrow \{0,1\}$ be any Boolean function.
Let $a$ denote the initial success probability of~$\mathcal A$.
Algorithm \textbf{algqsearch} finds a good solution using an expected number
of applications of $\mathcal A$ and ${\mathcal A}^{-1}$ which are in
$\Theta(\sqrt a)$ if $a>0$, and otherwise runs forever.}

\paragraph{} 

\begin{proof}[Proof of Theorem 1]
Suppose that \( A \cap B \neq \emptyset \) then, the support of \( \ket{\psi_x} \otimes \ket{\psi_y} \) contain a state \( \ket{\phi} \) which satisfies \(C\), or in other words $a = |\braket{\Psi_1|\Psi}|^2 > 0 $ and therefore by \textit{Theorem 3} there is an explicit procedure which takes a $\Theta(\sqrt{a})$ time in expectation, Hence for any \(\varepsilon >0\) we could construct a finite algorithm that fails with probability less than $ \varepsilon $ by rejecting runs that last longer than $\frac{1}{\varepsilon}$. 
  
On the other hand, Consider the case when \(A \cap B = \emptyset\) then $\Rightarrow a = 0 \Rightarrow \mathcal{H}_{\Psi}$ is 1-dimension space spanned only by $\ket{\Psi_0} $, and the operator $ I - 2 \ket{\Psi}\bra{\Psi} $ act over the $ \ket{\Psi_0}  $ as identity and therefore after executing any number of iterations the probability to measure from $\ket{\Psi_0}$ will remain $1$.

\paragraph{}Summarize the above yields the following protocol,
\begin{enumerate}
    \item Bob create \( \ket{\psi_x} \) and send it to Alice.
    \item Alice simulate \textbf{algqsearch} either the algorithm accept or either $n^4$ turns were passed.     
    \item If the algorithm accepts, Alice returns True; otherwise, Alice returns False. 
\end{enumerate}

The protocol computes the disjointness in a single round while requiring transmission of less than $\Theta\left( \sqrt{n} \right)$ qubits. That is in contrast to the known lower bound proved by Braverman \cite{Braverman}: 
\begin{theorem*}[Theorem A] The $r$-round quantum communication complexity of Disjointness$_n$ is $ \Omega\left( \frac{n}{r \log^8 r} \right)$.
\end{theorem*}
\paragraph{}
\end{proof}

\paragraph{Conclusion And Open Problems.} The reduction above demonstrate how known results can give us almost immediate insights into quantum compatibility. Besides being a no-go-to proof, we hope this work will also use as a hint for direction to other quantum advantages in the disturbed computing setting. 

It's worth saying that the $r$-rounds communication bound on disjointness does not hold in many cases. For a simple example, consider that each set $x,y \in \{0,1\}^{n}$ is drawn uniformly. Then it's clear that Alice and Bob could answer ''Yes'' and they will be correct with high probability. So the family of states, which one can project over them by only partly projection (diffusion operators), correspond to the distributions over pairs of Alice and Bob sets, which they can compute with their disjointness with less communication.        
\printbibliography 
\end{document}